\documentclass[sigconf]{acmart}
\usepackage[utf8]{inputenc}
\usepackage{tikz}
\usepackage{algorithm}
\usepackage{algorithmic}
\usepackage{subcaption}
\usepackage{amsmath}
\usepackage{amsmath}
\usepackage{multirow}
\usepackage{tcolorbox}
\usetikzlibrary{shapes.geometric, arrows.meta, positioning, fit, backgrounds, calc}

\usepackage{enumitem}
\newcommand{\partitle}[1]{%
  \par\smallskip
  \noindent\textbf{#1.}
}

\AtBeginDocument{%
  }

\setcopyright{acmlicensed}
\copyrightyear{2018}
\acmYear{2018}
\acmDOI{XXXXXXX.XXXXXXX}
\acmConference[Conference acronym 'XX]{Make sure to enter the correct
  conference title from your rights confirmation email}{June 03--05,
  2018}{Woodstock, NY}
\acmISBN{978-1-4503-XXXX-X/2018/06}

\begin{document}

\title{PriDyG: Privacy-preserving Dynamic Graph Inference with LLM-GNN Collaboration}

\author{Yuyang Xia}
\affiliation{%
  \institution{Emory University}
  \city{Atlanta}
  \state{Georgia}
  \country{USA}
}
\email{yuyang.xia@emory.edu}

\author{Ruixuan Liu}
\affiliation{%
  \institution{Emory University}
  \city{Atlanta}
  \state{Georgia}
  \country{USA}
}
\email{ruixuan.liu2@emory.edu}

\author{Li Xiong}
\affiliation{%
  \institution{Emory University}
  \city{Atlanta}
  \state{Georgia}
  \country{USA}
}
\email{lxiong@emory.edu}

\begin{abstract}
Graph inference over relational data can expose sensitive edge information, and this risk becomes more severe in dynamic graphs, where repeated model updates cause privacy loss to accumulate. We formulate Edge-level Differentially Private Dynamic Graph Inference (EDG) and propose PriDyG, a private inference framework that combines GNN-based structural learning with LLM-based semantic reasoning. PriDyG introduces incremental private multi-hop aggregation, which buffers newly arrived edges and processes each edge exactly once. By parallel composition, the total privacy cost equals that of a single static release, independent of the number or schedule of model updates. Compared with geometrically decaying budget allocation, incremental aggregation avoids exponentially increasing noise while preserving exact one-hop signals and at least half of two-hop information transfers. PriDyG further complements privatized GNN outputs with LLM predictions derived solely from node text, incurring no additional edge-level privacy cost. Experiments on four benchmarks for node classification and link prediction show that PriDyG consistently outperforms geometrically decaying baselines under the same privacy budget and matches the utility of naive per-update retraining while reducing cumulative privacy cost by up to three orders of magnitude.

\end{abstract}

\maketitle

\section{Introduction}

Graphs support many real-world applications, such as social network analysis, citation analysis, and e-commerce recommendation. Graph inference tasks~\cite{vert2004supervised,xu2018powerful,zhang2023inferturbo}, such as node classification and link prediction, are fundamental to these applications. However, graph structure inherently encodes sensitive information: edges may represent social connections, financial transactions, or medical consultations. Answering inference queries over such graphs may therefore expose sensitive relationships encoded in their structure, even when individual node attributes are considered non-sensitive~\cite{fu2023privacy}.

Differential privacy (DP)~\cite{dwork2014algorithmic} offers a principled framework for protecting sensitive information. In the graph setting, \textit{edge-level} DP guarantees that the presence or absence of any single edge cannot be reliably inferred from a mechanism's output~\cite{raskhodnikova2016differentially}. Recent work has made significant progress on privacy-preserving graph learning. GAP~\cite{sajadmanesh2023gap} and ProGAP~\cite{sajadmanesh2024progap} introduce aggregation perturbation to achieve edge-level DP for Graph Neural Networks (GNNs), while other works~\cite{zhang2024dpar,chien2023differentially,ran2024differentially,wu2023privacy} explore alternative privacy notions and mechanisms, including node-level DP. 
However, this protection comes at a steep utility cost (\textbf{Challenge 1}): enforcing edge-level DP requires injecting calibrated noise into the neighborhood aggregation process, which distorts precisely the structural signal that GNNs rely on. Under practical privacy budgets, DP-GNNs suffer substantial accuracy degradation relative to their non-private counterparts.

This challenge is further compounded when the graph evolves (\textbf{Challenge 2}). In practice, graphs are rarely static; they evolve as new events and interactions  occur~\cite{zeng2025efficient,sekara2016fundamental,sun2022aligning,mitra2025analyzing}, and inference models must be updated to reflect the current graph state. Continual inference on an evolving private graph raises a distinct privacy-accounting difficulty: as new edges arrive and the pipeline is repeatedly updated, historical edges are accessed again at every update, so privacy loss accumulates over time. Recent efforts have explored privacy protection for dynamic graphs. For example, Raskhodnikova et al.~\cite{raskhodnikova2025fully} design private algorithms for traditional tasks, such as triangle counting, on fully dynamic graphs, and DyGAN-EDP~\cite{ma2026novel} studies edge-level DP dynamic graph publishing by modeling a dynamic graph as a sequence of snapshots and synthesizing privacy-preserving future snapshots with a generative adversarial framework. However, these methods target private statistics or synthetic graph release rather than continual private inference over an evolving graph, and none of them addresses the accumulation of privacy loss across model updates. To the best of our knowledge, this is the first work to formulate and study edge-level DP graph inference over continually updated graphs, with the goal of ensuring that privacy cost does not grow with the number of graph updates.

To address these two challenges, we propose \textbf{PriDyG}, a privacy-preserving dynamic graph inference framework that couples LLM-based semantic reasoning with dynamic DP-GNN structural learning, as illustrated in Figure~\ref{fig:overview}. PriDyG runs two complementary branches: an LLM branch that answers queries from node textual attributes alone, and a DP-GNN branch that learns from the graph structure under edge-level DP; their predictions are combined through a confidence-gated fusion mechanism.
The key insight behind the collaborative design is that edge-level DP protects only the graph \emph{structure}: node-level textual attributes fall outside the protected data under this threat model, so LLM predictions computed solely from them consume \emph{zero} privacy budget. This privacy-free semantic signal compensates for the utility loss incurred by the privatized GNN, directly narrowing the privacy--utility gap of Challenge 1.

\begin{figure}[ht]
    \centering
\includegraphics[width=0.95\linewidth]{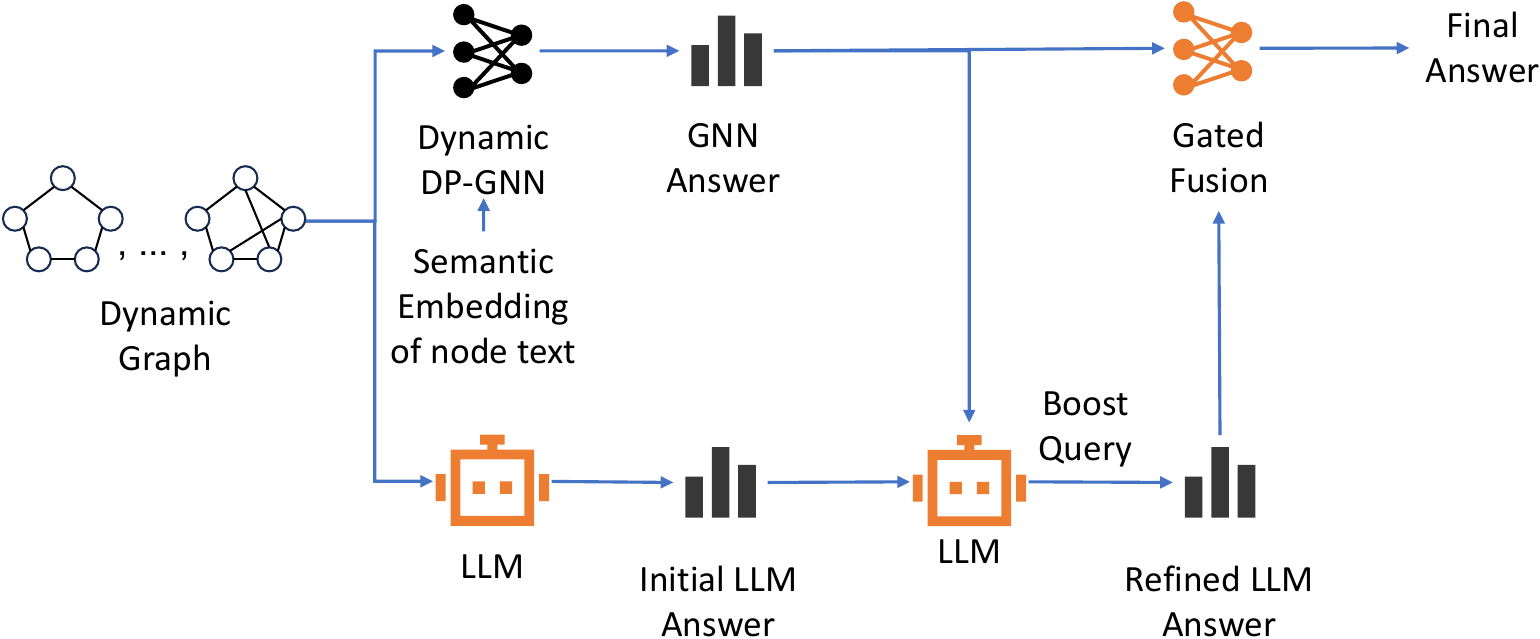}
    \caption{Overall architecture of PriDyG.}
    \label{fig:overview}
\end{figure}

For the dynamic DP-GNN branch, we build on aggregation-perturbation methods~\cite{sajadmanesh2023gap}, whose privacy cost is dominated by the Private Multi-hop Aggregation (PMA) step that perturbs the multi-hop neighborhood aggregation over the entire graph. Naively, every model update re-invokes PMA on the full graph; after $M$ updates, sequential composition yields a total cost of $M \cdot \epsilon_{\text{PMA}}$, which grows linearly and quickly becomes prohibitive. A seemingly viable alternative, following continual-release practice~\cite{feng2024dpi}, is to keep re-running full PMA while allocating a geometrically decaying budget to successive updates: the total cost then stays bounded, but the per-update budget decays exponentially toward zero, so the required noise scale grows exponentially and later models rapidly become unusable.

Our key observation is that the pre-normalization aggregation is \emph{linear} in the adjacency matrix, so newly inserted edges can be folded into the cached noisy aggregations as an additive delta---each edge's raw topology is accessed exactly once, and the \emph{parallel composition} theorem of DP bounds the total privacy cost by a \textbf{constant} $\epsilon_{\text{PMA}}$, independent of the number of updates $M$. The cost of this constant budget is paid in signal rather than privacy: the incremental aggregation only approximates a full retrain, omitting certain multi-hop propagation paths and accumulating noise across updates. We prove that this approximation is structured and bounded---one-hop aggregation is signal-exact, higher hops retain provable capture guarantees, and noise grows only \emph{polynomially}, in contrast to the exponential blow-up of geometric decay---resolving Challenge 2.
 
Importantly, the two components are coupled rather than independent fixes to separate problems. Node textual attributes remain fixed as the graph evolves, so the LLM branch provides a \emph{stable, privacy-free semantic anchor} throughout the dynamic process, whereas the accuracy of the incremental GNN branch gradually erodes as noise accumulates across updates. The confidence-gated fusion mediates between the two adaptively. As accumulating noise shrinks the GNN's prediction margins, more queries fall below the confidence threshold and are routed to fusion. We observe this in both node classification and link prediction tasks: over a long insertion stream the GNN branch alone causes substaintial accuracy drop relative to static training, whereas PriDyG's accuracy remains stable. Hence, the pipeline's reliance naturally shifts toward the stable LLM anchor over the course of the stream. The LLM-GNN collaboration is therefore most valuable precisely in the dynamic regime---where the structural branch is at its weakest---rather than being a generic accuracy boost bolted onto a dynamic mechanism.
 
\noindent Our main contributions are as follows:
\begin{itemize}[leftmargin=*]
    \item \textbf{LLM-GNN collaborative framework.} We design PriDyG, a unified pipeline that combines DP-GNN structural learning with privacy-free LLM semantic reasoning via confidence-gated fusion. The LLM branch compensates for the utility loss caused by privacy protection, thereby narrowing the privacy--utility gap.
 
    \item \textbf{Constant-budget dynamic graph inference.} We propose incremental PMA, which exploits the disjointness of inserted edge batches to achieve $O(1)$ total privacy cost under arbitrarily many model updates. This fundamentally resolves the privacy budget tension in dynamic graph inference, allowing frequent model updates without additional privacy degradation.
 
    \item \textbf{Comprehensive evaluation.} We evaluate PriDyG on four benchmark datasets (Cora, PubMed, ogbn-arxiv, ogbn-products) for node classification and link prediction. Under graph updates it is the most accurate private method on all four datasets while maintaining formal edge-level DP guarantee, and remains competitive with non-private baselines.
\end{itemize}

\section{Preliminaries}

We introduce the key concepts underlying our work and then formally define the problem. Table \ref{tab:notations.} summarizes important notations used in our paper.

\begin{table}[ht]
    \centering
    \scriptsize
    \begin{tabular}{ll}
        Notations & Definitions \\
        \hline
        $\mathcal{E}$ & Edge set\\
        $\mathbf{X}$ & Node feature \\
        $\mathbf{T}$ & Textual attribute\\
        $q$ & Graph query\\
        $\mathbf{A}$ & Adjacency matrix\\
        $\bar{\mathbf{Z}}$ & Normalized encodings for GNN\\
        $\mathcal{N}$ & Gaussian noise\\
        $\mathcal{B}$ & Edge buffer\\
        \hline
    \end{tabular}
    \caption{Notations.}
    \label{tab:notations.}
\end{table}

\subsection{Attributed and Dynamic Graphs}
Let $\mathcal{G} = (\mathcal{V}, \mathcal{E}, \mathbf{X}, \mathbf{T})$ denote an attributed undirected graph with $n$ nodes, edge set $\mathcal{E} \subseteq \binom{\mathcal{V}}{2}$, node features $\mathbf{X} \in \mathbb{R}^{n \times d}$, and textual attributes $\mathbf{T} = \{t_v\}_{v \in \mathcal{V}}$; each node carries a label $y_v \in \mathcal{Y} = \{1, \dots, C\}$. A \emph{dynamic graph} evolves through a stream of \textbf{edge insertions}: arriving edges accumulate in a buffer and are incorporated when a retraining command is issued, starting from an initial state $\mathcal{G}^{(0)}$ (e.g., containing only edges between training nodes). At any point, two types of inference queries may be issued: a \textbf{node classification} query $q_{\text{nc}}(v)$ predicts $\hat{y}_v \in \mathcal{Y}$, and a \textbf{link prediction} query $q_{\text{lp}}(u, v)$ predicts a score $s(u, v) \in [0, 1]$ for the likelihood that edge $\{u,v\}$ exists.

\subsection{Differential Privacy}
\label{sec:prelim-dp}
A randomized mechanism $\mathcal{M}: \mathcal{D} \to \mathcal{R}$ satisfies $(\epsilon, \delta)$-DP~\cite{dwork2014algorithmic} if for all neighboring $D, D' \in \mathcal{D}$ and all measurable $S \subseteq \mathcal{R}$, $\Pr[\mathcal{M}(D) \in S] \leq e^{\epsilon} \Pr[\mathcal{M}(D') \in S] + \delta$. We use the following standard facts.

\noindent\textbf{Edge-level DP.} Two graphs are \emph{neighbors} if they differ by a single undirected edge, $|\mathcal{E} \triangle \mathcal{E}'| = 1$; DP under this adjacency protects the existence of any individual edge. Removing one undirected edge $\{u,v\}$ changes the neighbor aggregation of \emph{both} endpoints, each by at most one unit-norm vector, so the $\ell_2$-sensitivity of aggregation is $\sqrt{2}$ (vs.\ $1$ under the directed-edge convention of GAP~\cite{sajadmanesh2023gap}).

\noindent\textbf{Gaussian mechanism and accounting.} For $f$ with $\ell_2$-sensitivity $\Delta_2(f)$, releasing $f(D) + \mathcal{N}(0, \sigma^2 \mathbf{I})$ satisfies $(\epsilon,\delta)$-DP with $\sigma$ calibrated by a standard accountant; we use the R\'enyi-DP accountant~\cite{mironov2017renyi} of GAP~\cite{sajadmanesh2023gap}. Our analysis is accountant-agnostic (Section~\ref{sec:privacy-analysis}).

\noindent\textbf{Parallel composition.} If $D$ is partitioned into disjoint subsets $D_1, \dots, D_M$ and $\mathcal{M}_m$ accesses only $D_m$ with cost $\epsilon_m$, the total cost of releasing all outputs is $\max_m \epsilon_m$, not the sum~\cite{smith2021making}. This is the structural tool behind our constant budget: naive retraining accesses all edges each time (sequential composition, linear growth), whereas incremental PMA partitions edges into disjoint batches across retrains.

\noindent\textbf{Post-processing immunity.} Any function of a DP release satisfies the same guarantee; downstream computations on privatized outputs (classifier training, fusion, LLM use of released scores) incur no additional cost.

\subsection{Edge-Level DP Graph Learning}
\label{sec:prelim-gap}

Our GNN module builds directly on GAP~\cite{sajadmanesh2023gap}, whose decoupled three-stage pipeline---encoder pre-training, private multi-hop aggregation, and classifier training---we recall here; the construction in this subsection is entirely from~\cite{sajadmanesh2023gap}, except for the sensitivity constant noted below.

\paragraph{Encoder.}
An MLP encoder $f_\theta$ maps node features to lower-dimensional
representations $\mathbf{z}_v = f_\theta(\mathbf{x}_v)$. Since the
encoder accesses only node features and labels, never edges, this stage
is privacy-free under edge-level DP.

\paragraph{Private multi-hop aggregation (PMA).}
Let $\mathbf{A}$ be the adjacency matrix of the current graph and let
$\bar{\mathbf{Z}}^{(0)}$ be the $\ell_2$-row-normalized encodings. For
each hop $k = 1, \dots, K$:
\begin{equation}
\label{eq:pma-agg}
    \tilde{\mathbf{Z}}^{(k)} = \mathbf{A} \, \bar{\mathbf{Z}}^{(k-1)} + \mathcal{N}(\mathbf{0}, \sigma^2 \mathbf{I}),
    \qquad
    \bar{\mathbf{Z}}^{(k)} = \text{Normalize}(\tilde{\mathbf{Z}}^{(k)}),
\end{equation}
where $\text{Normalize}(\cdot)$ applies row-wise $\ell_2$-normalization. Because every row of $\bar{\mathbf{Z}}^{(k-1)}$ has unit norm, the aggregation in~\eqref{eq:pma-agg} has $\ell_2$-sensitivity $\Delta_2 = \sqrt{2}$ under the undirected-edge adjacency of Section~\ref{sec:prelim-dp}. The Gaussian mechanism applied over the $K$ hops therefore yields an edge-level $(\epsilon_{\mathrm{PMA}}, \delta)$-DP guarantee, where $\epsilon_{\mathrm{PMA}}(\sigma, \delta)$ denotes the cost of a single $K$-hop PMA invocation.

\paragraph{Classifier.}
A classifier $g_\phi$ takes the multi-hop features
$\{\bar{\mathbf{Z}}^{(0)}, \dots, \bar{\mathbf{Z}}^{(K)}\}$ as input and predicts node labels. As it operates on already-privatized features, training it incurs no additional privacy cost by post-processing
immunity.

\subsection{Problem Formulation}

\begin{definition}[EDG Problem]
\label{def:edg}
Given a dynamic graph $\{\mathcal{G}^{(0)}, \mathcal{G}^{(1)}, \dots\}$ evolving through edge insertions and a stream of queries $\{q_1, q_2, \dots\}$, where each query $q_i$ is either a node classification query $q_{\text{nc}}(v)$ or a link-prediction query $q_{\text{lp}}(u,v)$, our goal is to design a mechanism $\mathcal{M}$ that (1)~answers each query using the current graph state; (2)~satisfies edge-level $(\epsilon, \delta)$-DP over the \emph{entire execution}; and (3)~maintains high prediction accuracy and ranking quality despite privacy constraints and graph evolution.
\end{definition}

Our threat model protects only edge existence; node text and node labels are treated as public. Under this assumption, LLM predictions derived solely from node text incur no additional edge-level privacy cost.
The central challenge is the \emph{privacy budget tension}: each retraining consumes budget under sequential composition, growing as $M \cdot \epsilon_{\text{PMA}}$ after $M$ retrains. Our incremental PMA processes each edge exactly once, so parallel composition applies and $\epsilon_{\text{total}} = \epsilon_{\text{PMA}}$, independent of $M$ and of the retraining schedule.

\section{Methodology}
\label{sec:method}
In this section, we present PriDyG, a privacy-preserving dynamic graph inference pipeline that combines GNN-based structural learning with LLM-based semantic reasoning under edge-level differential privacy. 
Throughout this section, we present PriDyG using node classification as the default task; Section \ref{sec:link_prediction} describes how the framework is adapted to link prediction. Figure \ref{fig:overview} illustrates the overall architecture.

\subsection{Overview}

PriDyG operates in a streaming setting where edges arrive over time and inference queries must be answered on the fly. The pipeline consists of three modules:
\begin{enumerate}[leftmargin=*, nosep]
    \item \textbf{GNN module}: A DP GNN based on the GAP framework~\cite{sajadmanesh2023gap} that learns structural representations through private multi-hop aggregation (PMA).
    \item \textbf{LLM module}: An LLM that performs zero-shot classification using node textual attributes, consuming no privacy budget.
    \item \textbf{Fusion module}: A confidence-gated mechanism that dynamically combines GNN and LLM predictions.
\end{enumerate}

Since the LLM operates solely on node-level textual attributes,
its predictions are a function of data outside the protected
domain: edge-level DP constrains only the edge set, while node
text is public in our setting {(Definition~\ref{def:edg})}.
LLM inference therefore consumes zero privacy budget, allowing
PriDyG to maintain high-quality predictions even when the GNN
module operates under a tight privacy budget.

\subsection{Node Classification}

\partitle{GNN Module}
The GNN module is the edge-level DP pipeline of
Section~\ref{sec:prelim-gap}. The encoder
maps node text embeddings to $\bar{\mathbf{Z}}^{(0)}$; $K$ rounds of
Gaussian-perturbed aggregation~\eqref{eq:pma-agg} over the current graph
state yield the privatized multi-hop features
$\{\bar{\mathbf{Z}}^{(k)}\}_{k=0}^{K}$ at an edge-level cost of
$\epsilon_{\mathrm{PMA}}$; and the classifier $g_\phi$ maps them to
label logits for a query node $q$:
\begin{equation}
\label{eq:gnn-logits}
    \hat{\mathbf{y}}_q = g_\phi\!\left(\bar{\mathbf{Z}}^{(0)}_q, \dots, \bar{\mathbf{Z}}^{(K)}_q\right) \in \mathbb{R}^{C}.
\end{equation}
The module exposes two quantities downstream: the logit vector
$\hat{\mathbf{y}}_q$, from which the fusion step below forms
$\mathbf{s}_q^{\text{graph}}$, and the top-$m$ labels induced by
$\hat{\mathbf{y}}_q$, which seed the second-stage candidate set of the
LLM module. Since $\{\bar{\mathbf{Z}}^{(k)}\}_{k=0}^{K}$ is a DP
release, both are post-processing and incur no cost beyond
$\epsilon_{\mathrm{PMA}}$.

\partitle{LLM Module: Semantic Reasoning}
For a query node $q$ with textual attributes $t_q$, the LLM computes a probability distribution over candidate labels via token log-probabilities:
\begin{equation}
    P_{\text{LLM}}(l \mid t_q) \propto \exp\!\left(\sum_{j=1}^{|l|} \log P(l_j \mid l_{<j}, \text{prompt})\right).
\end{equation}
where $l_j$ denotes the $j$-th token of label $l$, and $\mathrm{prompt}$ contains the query node's textual attributes and the candidate label set. We employ a two-stage boost query as shown in Figure \ref{fig:overview}: the first stage scores all $C$ candidate labels; the second stage re-scores a smaller candidate set formed by the union of top-$m$ labels from the GNN score and the LLM score, so the total size of the union is $2r$. This strategy shrinks the context window and helps the LLM reduce the prompt length and focus on the most likely labels. Since the LLM accesses only public node text (not edges) and DP GNN predictions, it introduces no additional privacy cost.

\partitle{Confidence-Gated Fusion}
Let $\mathbf{s}_q^{\text{graph}} = \text{softmax}(\hat{\mathbf{y}}_q)$ and $\mathbf{p}_q^{\text{LLM}}$ denote the GNN and LLM score vectors, respectively. Since invoking the LLM for every query incurs substantial computational overhead, we adopt confidence-gated fusion to selectively use the LLM only when the graph-based prediction is uncertain. This design reduces unnecessary LLM inference while allowing complementary textual information to assist with ambiguous queries. Specifically, we define the GNN confidence margin as $m_q = s_q^{(1)} - s_q^{(2)}$, where $s_q^{(1)}$ and $s_q^{(2)}$ are the largest and second-largest probabilities, respectively. If $m_q > \tau$, we directly use the GNN prediction. Otherwise, we fuse the GNN and LLM predictions over the top-$2m$ candidates:
\begin{equation}
    \mathbf{p}_q^{\text{fused}}[l] = \alpha \cdot \bar{\mathbf{s}}_q^{\text{graph}}[l] + (1 - \alpha) \cdot \bar{\mathbf{p}}_q^{\text{LLM}}[l], \quad l \in \text{top-}2m,
\end{equation}
where $\bar{\mathbf{s}}$ and $\bar{\mathbf{p}}$ are re-normalized to the top-$2m$ set. The parameters $(\alpha, \tau)$ are tuned via grid search on a validation set. Since node labels are treated as public under our edge-level privacy model, this tuning incurs no additional privacy cost.

\subsection{Link Prediction
}
\label{sec:link_prediction}
We now turn to the second query type of Definition \ref{def:edg}. 
Link prediction is particularly delicate under edge-level DP, because the prediction target is the existence of an edge, which is exactly the information to be protected. 
In node classification, labels are public supervision, so that channel is privacy-free; any supervised link predictor instead reads edge-existence labels and thus protected data. PriDyG therefore stays in post-processing: its default link predictor is a training-free function of the privatized PMA representations, and the only component that requires edge labels is calibrated on randomized-response labels on a disjoint set of node pairs.

\partitle{Cosine Scoring as the Default Link Predictor}
After private multi-hop aggregation, each node $v$ obtains a privatized multi-hop representation
$
\bar{\mathbf z}_v =
[\bar{\mathbf z}^{(0)}_v \| \bar{\mathbf z}^{(1)}_v \| \cdots \| \bar{\mathbf z}^{(K)}_v],
$
where $\bar{\mathbf z}^{(k)}_v$ is derived from PMA. For a candidate node pair $(u,v)$, PriDyG computes a structural link score as the  cosine similarity of these representations:
$
s_{\mathrm{cos}}(u,v)
=
\frac{
\bar{\mathbf z}_u^\top \bar{\mathbf z}_v
}{
\|\bar{\mathbf z}_u\|_2 \|\bar{\mathbf z}_v\|_2
}.
$
For fusion and calibration, we map the cosine score to $[0,1]$ by the deterministic transformation
$
\tilde{s}_{\mathrm{cos}}(u,v)
=
\frac{s_{\mathrm{cos}}(u,v)+1}{2}.
$

This cosine-based predictor is the default link prediction method in PriDyG. It does not access raw edge labels and does not inspect the original graph during inference. Since it only operates on already privatized PMA representations, it is a post-processing step of the edge-level DP PMA mechanism and incurs no additional privacy cost.

\partitle{LLM semantic score}
Cosine scoring captures privatized structural similarity, but DP noise perturbs every $\bar{\mathbf z}_v$ and can damage the fine-grained ordering that ranking metrics depend on. As complementary evidence, PriDyG computes a semantic score from node text alone:
$
s_{\mathrm{LLM}}(u,v)
=
P_{\mathrm{LLM}}(\text{``Yes''} \mid t_u, t_v),
$
where $t_u$ and $t_v$ are the textual attributes of nodes $u$ and $v$. The prompt is instantiated per domain: for citation graphs, it asks whether two papers are likely connected by citation or topical relevance; for product graphs, whether two products are likely co-purchased.
Since the prompt contains only public node text and no topology, 
$s_{\mathrm{LLM}}$ consumes no privacy budget.

\partitle{Private fusion-weight calibration}
The two scores are combined by the confidence-gated rule of node classification, with the gate measured on a scalar score.
The fusion is post-processing, so the choice of $(\alpha, \tau)$ is not free.
Unlike node classification, where these are tuned on public node labels, the tuning signal is the edge existence itself, which edge-level DP protects.
We therefore calibrate on randomized-response labels.

We randomly partition the candidate pairs in $\binom{\mathcal{V}}{2}$ into the fusion partition 
$\mathcal{P}_{fus}$ (with $r = 5\%$ assigned to $\mathcal{P}_{fus}$) and the remainder $\mathcal{P}_{train}$.
At $t=0$, for every pair $(u,v)\in\mathcal{P}_{fus}$, we release a single privatized copy of its existence bit $b_{uv}=\mathbb{1}[(u,v)\in\mathcal{E}^{(0)}]$ via randomized response~\cite{wang2016using},
\begin{equation}
\tilde{b}_{uv}=b_{uv} \text{ w.p. } 
\tfrac{e^{\epsilon_{\mathrm{RR}}}}{e^{\epsilon_{\mathrm{RR}}}+1},
\quad
\tilde{b}_{uv}=1-b_{uv} \text{ otherwise},
\label{eq:rr}
\end{equation}
where $\epsilon_{\mathrm{RR}}$ is the privacy budget allocated to randomized response. 

Parallel composition applies only if each edge is touched by one mechanism, so we exclude the
initial edges in $\mathcal{P}_{fus}$ from the initial
PMA adjacency. 
A pair of $\mathcal{P}_{fus}$ that becomes an edge at some $t>0$ is instead processed by incremental PMA (Section \ref{sec:incremental_pma}). %
Hence, initial calibration edges are protected through randomized response, while subsequent edge insertions are protected through incremental PMA. 
The weight $\alpha$ and threshold $\tau$ are then selected by grid search, maximizing the
empirical AUC computed on the privatized link labels
$\{\tilde{b}_{uv}\}$.
\begin{lemma}[Privacy of fusion calibration]
\label{lem:rr}
The release of $\{\tilde{b}_{uv}\}_{(u,v)\in\mathcal{P}_{\mathrm{fus}}}$
satisfies edge-level $(\epsilon_{\mathrm{RR}},0)$-DP. Moreover,
the joint release of the full PriDyG output, incremental PMA
outputs (Theorem~\ref{thm:constant-budget}) and the calibration
bits, satisfies edge-level
$\bigl(\max(\epsilon_{\mathrm{PMA}},\epsilon_{\mathrm{RR}}),\,\delta\bigr)$-DP,
for any number of retrains and any schedule.
\end{lemma}
The proof of Lemma~\ref{lem:rr} is in Appendix~\ref{appendix:proofs}.

In all experiments, we set
$\epsilon_{\mathrm{RR}}=\epsilon_{\mathrm{PMA}}$, so the
end-to-end guarantee of PriDyG remains
$(\epsilon_{\mathrm{PMA}},\delta)$-DP, unchanged from
Theorem~\ref{thm:constant-budget}.

\subsection{Dynamic Graph Handling via Incremental PMA}\label{sec:incremental_pma}

A central challenge in dynamic graph inference under DP is that structural changes require re-running PMA, yet each invocation consumes privacy budget. Naively re-running PMA on the \emph{entire} graph at every retrain would cause the total privacy cost to grow linearly as $M \cdot \epsilon_{\text{PMA}}$, where $M$ is the number of retrains. We introduce \emph{incremental PMA}, which keeps the total privacy budget at a \textbf{constant} $\epsilon_{\text{PMA}}$ regardless of the number of retrains.

\subsubsection{Baseline: Geometrically Decaying Budget Allocation}
\label{sec:geometric}
A straightforward baseline is to rerun full PMA over the entire
current graph whenever a retraining command is issued, while assigning
a geometrically decaying privacy budget to successive retraining
operations. Specifically, given a total privacy budget
$\epsilon_{\mathrm{tot}}$, the budget allocated to the $m$-th
retraining can be defined as
$
\epsilon_m = (1-q)q^m \epsilon_{\mathrm{tot}},
0<q<1.
$
Since
$
\sum_{m=0}^{\infty} \epsilon_m = \epsilon_{\mathrm{tot}},
$
this strategy supports an unbounded number of retraining operations
without exceeding the prescribed total privacy budget. Similar
geometric budget-allocation strategies have been adopted for
continual DP release~\cite{feng2024dpi}.

However, this baseline suffers from severe utility degradation over
long-running update streams. As the number of retraining operations
increases, the privacy budget available to later retraining operations
decays exponentially toward zero. Consequently, the required noise
grows rapidly, making the updated model increasingly inaccurate.
Moreover, because each retraining still accesses the entire graph,
historical edges are repeatedly processed; the method controls the
cumulative privacy loss only by progressively weakening the utility
of later releases.

\subsubsection{Incremental PMA Algorithm}
\label{sec:incremental-pma-alg}

The graph evolves through edge insertions. 
As shown in Figure \ref{fig:incremental pma}, incremental PMA maintains an \emph{edge buffer} $\mathcal{B}$ that accumulates newly arrived edges since the last retrain. When a retraining command is issued (by any external policy), incremental PMA processes all edges in $\mathcal{B}$, updates the cached aggregations, and clears $\mathcal{B}$. We denote the edges processed at the $m$-th retrain as $\Delta\mathcal{E}^{(m)}$ and write $\mathcal{E}^{(m)} = \mathcal{E}^{(m-1)} \cup \Delta\mathcal{E}^{(m)}$. The retraining schedule itself is not part of the algorithm and does not affect the privacy guarantee.

At the initial retrain ($m=0$), we run full PMA on $\mathcal{E}^{(0)}$
and cache both the pre-normalization aggregations
$\tilde{\mathbf{Z}}^{(k)}$ and their row-normalized counterparts
$\bar{\mathbf{Z}}^{(k)}$ for all hops $k$. At each subsequent retrain
$m\ge 1$, the hops are updated \emph{sequentially, in increasing order}
$k=1,2,\dots,K$; the update for hop $k$ begins only after the update
for hop $k-1$ has completed. Each retrain proceeds as follows:
\begin{enumerate}

\item\textbf{Build delta adjacency matrix.} Construct $\Delta\mathbf{A}^{(m)}$
from the new edges $\Delta\mathcal{E}^{(m)}$.

\item\textbf{Compute delta aggregation with noise.} For each hop
$k=1,\dots,K$ \emph{in order}:
\begin{equation}
\Delta\tilde{\mathbf{Z}}^{(k)}
=\Delta\mathbf{A}^{(m)}\,\bar{\mathbf{Z}}^{(k-1)}_{\mathrm{new}}
+\mathcal{N}(\mathbf{0},\sigma^2\mathbf{I}),
\label{eq:delta-agg}
\end{equation}
where $\bar{\mathbf{Z}}^{(k-1)}_{\mathrm{new}}$ denotes the hop-$(k{-}1)$
normalized features \emph{finalized in Step 3 of the current retrain},
i.e., strictly before hop $k$ is processed. For the base case $k=1$,
$\bar{\mathbf{Z}}^{(0)}_{\mathrm{new}}=\bar{\mathbf{Z}}^{(0)}$ is the
row-normalized encoder output, which depends only on node text and is
therefore identical across all retrains.
\item \textbf{Update the cached features.}
\begin{equation}
\tilde{\mathbf{Z}}^{(k)}_{\mathrm{new}}
=\tilde{\mathbf{Z}}^{(k)}_{\mathrm{cached}}+\Delta\tilde{\mathbf{Z}}^{(k)},
\qquad
\bar{\mathbf{Z}}^{(k)}_{\mathrm{new}}
=\mathrm{Normalize}\bigl(\tilde{\mathbf{Z}}^{(k)}_{\mathrm{new}}\bigr),
\end{equation}
where $\tilde{\mathbf{Z}}^{(k)}_{\mathrm{cached}}$ is the
pre-normalization aggregation carried over from the previous PMA.
 Store $\tilde{\mathbf{Z}}^{(k)}_{\mathrm{new}}$
and $\bar{\mathbf{Z}}^{(k)}_{\mathrm{new}}$ for the next retrain.

\end{enumerate}
\textbf{Sequential hop updates}.
\label{rem:no-circularity}
Equation~\eqref{eq:delta-agg} intentionally multiplies
$\Delta\mathbf{A}^{(m)}$ by the \emph{updated} lower-hop features
$\bar{\mathbf{Z}}^{(k-1)}_{\mathrm{new}}$ rather than by the stale
cache $\bar{\mathbf{Z}}^{(k-1)}_{\mathrm{cached}}$. This introduces no
circularity: within a retrain, hop $k$ consumes only the
\emph{finalized} output of hop $k-1$, exactly mirroring the sequential
hop structure of full PMA in
Eqs.~\eqref{eq:pma-agg}, and hop $k$ never depends
on its own output. The choice is also important for utility: using the
stale cache would let new edges aggregate only pre-update features, so
a two-hop path traversing two newly inserted edges would transfer no
information at all. With Eq.~\eqref{eq:delta-agg}, hop $1$ is exact
(since $\bar{\mathbf{Z}}^{(0)}$ is edge-independent,
$\tilde{\mathbf{Z}}^{(1)}_{\mathrm{new}}$ equals the hop-$1$
aggregation of a full retrain up to the noise realization), and higher
hops satisfy the capture guarantees of
Section~\ref{sec:error-analysis}. Privacy is unaffected by this adaptivity: as shown in the proof of
Theorem~\ref{thm:constant-budget}, conditional on the outputs released
so far, each hop is a Gaussian mechanism of sensitivity $\sqrt{2}$, and
the $K$ hops of a retrain have exactly the same mechanism structure as
one static PMA invocation, so any accounting valid for static PMA
applies unchanged.
 
\noindent\textbf{Which modules are retrained.}
\label{rem:frozen-modules}
Incremental PMA caches
aggregations of \emph{fixed} input features. Accordingly, after the
initial retrain, the encoder $f_\theta$ is frozen; only the
classification head $g_\phi$, whose output is never re-aggregated, is
retrained at each retrain, and since $g_\phi$ consumes only
already-privatized features, its retraining is post-processing and
incurs no privacy cost.

\begin{figure}[ht]
    \centering
    \includegraphics[width=0.6\linewidth]{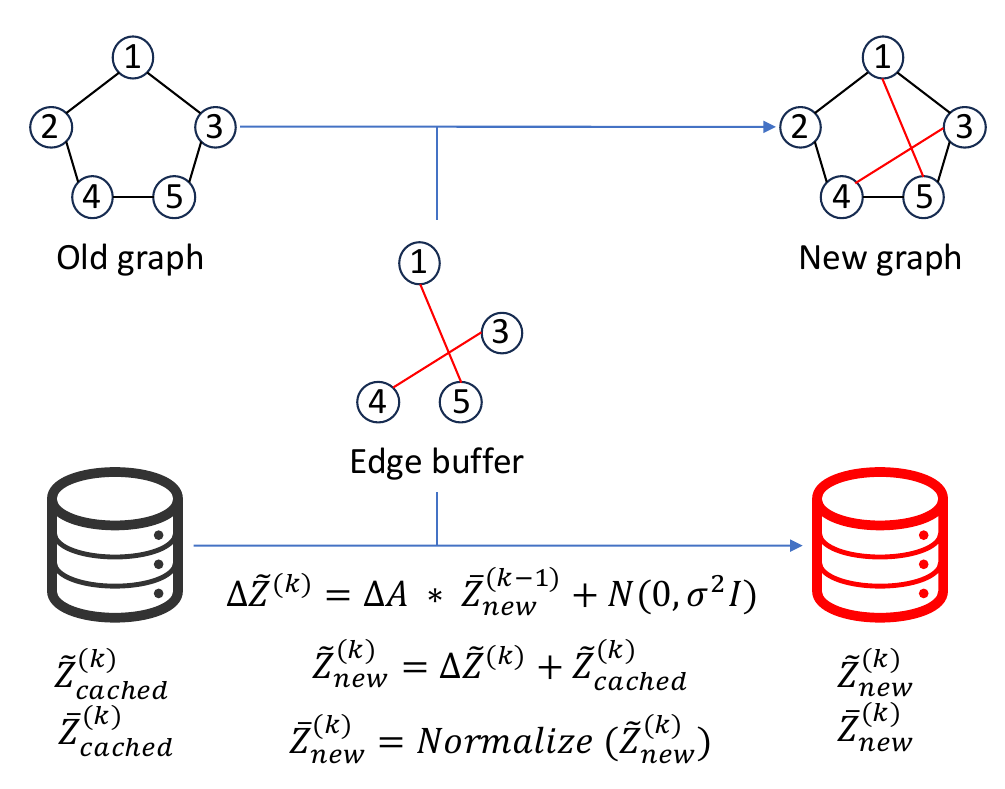}
    \caption{Example for incremental PMA.}
    \label{fig:incremental pma}
\end{figure}

\begin{example}
    \label{exp:incremental pma}
    As shown in Figure \ref{fig:incremental pma}, the edge buffer is $\{(1,5), (3,4)\}$.
    When the next retraining command is issued, incremental PMA builds the
delta adjacency $\Delta\mathbf{A}$, whose only nonzero entries correspond to the two undirected edges: 
$\Delta\mathbf{A}[1,5]=\Delta\mathbf{A}[5,1]=\Delta\mathbf{A}[3,4]=\Delta\mathbf{A}[4,3]=1$.
The delta aggregation at hop $k$ therefore touches only the four
endpoints of the buffered edges:
\[
\Delta\mathbf{A}\,\bar{\mathbf{Z}}^{(k-1)}_{\mathrm{new}}
=\bigl[\,\bar{\mathbf{Z}}^{(k-1)}_{\mathrm{new}}[5],\;
\mathbf{0},\;
\bar{\mathbf{Z}}^{(k-1)}_{\mathrm{new}}[4],\;
\bar{\mathbf{Z}}^{(k-1)}_{\mathrm{new}}[3],\;
\bar{\mathbf{Z}}^{(k-1)}_{\mathrm{new}}[1]\,\bigr].
\]
The entries correspond to node 1 through 5, respectively.  That is, node $1$ receives node $5$'s (already updated) hop-$(k{-}1)$
representation and vice versa, nodes $3$ and $4$ receive each other's representations, and
node $2$---incident to no new edge---receives the zero vector. The noisy
delta is then added to the cached pre-normalization features
$\tilde{\mathbf{Z}}^{(k)}_{\mathrm{cached}}$ and re-normalized.
   
\end{example}
\subsection{Privacy Analysis of Incremental PMA}
\label{sec:privacy-analysis}
The privacy guarantee of incremental PMA follows from a simple structural property: each edge is processed exactly once. When an edge arrives, it enters the buffer; when a retraining command is issued, the buffer is cleared. Regardless of the retraining schedule, each edge incurs the privacy cost of one static PMA invocation.  %
This decouples the privacy analysis from the retraining policy entirely.

\begin{theorem}[Constant privacy budget under incremental PMA]
\label{thm:constant-budget}
Fix the noise scale $\sigma$, the number of hops $K$, and
$\delta\in(0,1)$, and let $\epsilon_{\mathrm{PMA}}$ be any value such
that a single static $K$-hop PMA invocation with noise scale $\sigma$
satisfies edge-level $(\epsilon_{\mathrm{PMA}},\delta)$-DP (e.g., as
certified by the accountant of~\cite{sajadmanesh2023gap}). Let
$B_0=\mathcal{E}^{(0)}$ and $B_m=\Delta\mathcal{E}^{(m)}$ for
$m\ge 1$ denote the edge batches drained at each retrain. Then for any
number of retrains $M$ and any retraining schedule, the full output
of incremental PMA,
$
Y_{0:M}=(Y_0,Y_1,\dots,Y_M),
Y_m=\bigl\{\Delta\tilde{\mathbf{Z}}^{(k)}_m\bigr\}_{k=1}^{K},
$
satisfies edge-level $(\epsilon_{\mathrm{PMA}},\delta)$-DP. In
particular, the privacy cost does not grow with $M$, does not depend on
the retraining schedule, and equals that of a single static PMA
invocation.
\end{theorem}

Proof of Theorem \ref{thm:constant-budget} can be found in Appendix \ref{appendix:proofs}.

\subsection{Error Analysis of Incremental PMA}
\label{sec:error-analysis}
We now characterize the approximation error of incremental PMA relative to a full retrain on the current graph. There are two sources of errors.  First, because cached contributions from older edges are not recomputed, incremental PMA omits propagation paths that traverse a newer edge followed by an older edge. This yields a structured approximation error: one-hop propagation is exact, and at least half of all two-hop transfers are retained. Second, each retrain adds fresh Gaussian noise to the cached aggregation, causing the accumulated noise variance to grow linearly with the number of retrains. We formalize these two sources of error below.
We first analyze path capture in Theorem \ref{thm:captured-paths} using noiseless, pre-normalization linear propagation (row-wise
$\ell_2$-normalization is a per-hop post-processing and does not change
which path can contribute information). We then analyze noise accumulation 
separately in Proposition~\ref{prop:noise-accum}. 

\partitle{Captured-path analysis}
Assign each edge the
index of the retrain at which it was drained from the buffer, with the
initial edge set $\mathcal{E}^{(0)}$ as batch $0$, and let
$\mathbf{A}_i$ denote the adjacency matrix of $i$th batch $B_i$. Write
$\mathbf{Z}=\bar{\mathbf{Z}}^{(0)}$, which depends only on node text
and is identical across retrains, and let $\mathbf{F}^{(k)}$ and
$\mathbf{I}^{(k)}$ denote the $k$-hop aggregations after retrain $m$
under a full retrain and under incremental PMA, respectively.
 
\begin{theorem}[Captured-path characterization]
\label{thm:captured-paths}
After any retrain $m$ and for every hop $k$,
\begin{equation}
\mathbf{I}^{(k)}
=\!\!\sum_{0\,\le\, i_1\le\cdots\le i_k\,\le\, m}\!\!
\mathbf{A}_{i_k}\cdots\mathbf{A}_{i_1}\,\mathbf{Z},
\label{eq:captured-paths}
\end{equation}
whereas the full retrain
$\mathbf{F}^{(k)}=\bigl(\sum_{i\le m}\mathbf{A}_i\bigr)^{k}\mathbf{Z}$
expands into \emph{all} length-$k$ paths over
$\{\mathbf{A}_0,\dots,\mathbf{A}_m\}$. Equivalently, a $k$-hop path
contributes to $\mathbf{I}^{(k)}$ if and only if the batch indices of
its edges, read in traversal order, are non-decreasing: information may
flow from older edges onto newer ones, but never from newer edges back
across older ones.
\end{theorem}
 
Proof of Theorem \ref{thm:captured-paths} can be found in Appendix~\ref{appendix:proofs}, which unrolls the update
recursion across retrains; the error
$\mathbf{F}^{(k)}-\mathbf{I}^{(k)}$ collects exactly the words
containing an older-batch step \emph{after} a newer-batch step, whose
source is the staleness of the cache: old edges, frozen in the cache,
never re-aggregate features refreshed by later insertions. The error
is thus structured and bounded---incremental PMA omits only these
``backward-in-time'' paths---rather than an unbounded drift, and it
vanishes entirely for $K=1$:
 
\begin{corollary}[One hop: lossless in signal]
\label{cor:hop1-exact}
After every retrain, $\mathbf{I}^{(1)}=\mathbf{F}^{(1)}$, thus  incremental
PMA preserves all one hop information in the noiseless aggregation. %
\end{corollary}

\begin{corollary}[Two hops: at least half of all wedge transfers]
\label{cor:hop2-capture}
Two-hop information flows along wedges $u\!-\!w\!-\!v$, each carrying
two distinct information transfers, one per traversal direction. Let
$\rho$ be the fraction of wedge transfers whose two edges lie in
distinct batches. Then
$\text{captured fraction at hop } 2
= 1-\tfrac{1}{2}\rho \ge \tfrac{1}{2}.$

\end{corollary}
 
In Example~\ref{exp:incremental pma}, edges $(1,3)$ and $(1,5)$ arrive in
different batches, so the transfer $3\!\to\!1\!\to\!5$ is the captured
path and $5\!\to\!1\!\to\!3$ is the lost one.

\partitle{Noise accumulation analysis}
We next quantify how the Gaussian perturbations introduced across incremental retrains accumulate in the cached representations and compare this growth with the geometrically decaying baseline.
 
\begin{proposition}[Noise accumulation: linear vs.\ exponential]
\label{prop:noise-accum}
Fix the noise scale $\sigma$ in each PMA invocation. After $M$ incremental retrains, the
hop-$k$ cache $\tilde{\mathbf{Z}}^{(k)}$ contains the sum of $M{+}1$
independent Gaussian noise 
$\mathcal{N}(\mathbf{0},\sigma^2\mathbf{I})$, yielding a total variance of $(M{+}1)\sigma^2$ and a $\sqrt{M{+}1}$-fold inflation in
noise magnitude over a single full PMA on the final graph under the same 
privacy budget. At hop $1$,  Corollary~\ref{cor:hop1-exact} shows that the noiseless signal matches the one-shot full retrain, so this
noise is the only deviation. In
contrast, under the geometrically decaying baseline of
Section~\ref{sec:geometric}, the per-retrain noise scale grows
\emph{exponentially}, $\sigma_m=\Theta\!\bigl(q^{-m}\bigr)$ at fixed
$\delta$, and perturbs the entire released representation.
\end{proposition}
 
The one-shot reference is an \emph{oracle} that observes the
entire final graph in advance and therefore
cannot answer any query while the stream is evolving. So the
$\sqrt{M{+}1}$ overhead should be interpreted as the intrinsic cost of
continual release relative to the unattainable oracle, and price is an $O(M)$ accumulation of noise variance.
Incremental PMA therefore trades a structured, bounded approximation
error (omitting newer-to-older paths while retaining at least half of
all two-hop transfers) and a polynomially growing noise floor for a
constant privacy budget, whereas geometric decay is structurally
unbiased but exponentially noisy.

Proofs of Corollary \ref{cor:hop1-exact}, \ref{cor:hop2-capture}, and Proposition \ref{prop:noise-accum} can be found in Appendix \ref{appendix:proofs}.

\section{Experiments}
\label{sec:experiments}
We evaluate PriDyG on node classification and link prediction across four benchmark datasets under edge-level DP, in both static and dynamic settings. 

\subsection{Experimental Setup}
\partitle{Datasets}
We use four graph benchmarks spanning different domains and scales:
\textbf{Cora}~\cite{mccallum2000automating}: A citation network of machine learning papers, where each node represents a paper with title and abstract, and edges denote citation links. The task is to classify papers into 7 categories (e.g., Neural Networks, Reinforcement Learning). \textbf{PubMed}~\cite{xu2025pubmed}: A biomedical citation network of diabetes-related publications. Nodes are papers classified into 3 categories of diabetes research.
\textbf{ogbn-arxiv}~\cite{hu2020open}: A large-scale citation network of computer science arXiv papers spanning 40 subject areas (e.g., cs.CV, cs.LG, cs.CL).
\textbf{ogbn-products}~\cite{hu2020open}: A co-purchasing network from Amazon, where nodes are products with textual descriptions and edges connect frequently co-purchased items. Products are classified into 44 categories.
Each dataset is integrated with a train/valid/test node split. Therefore, each node in the graph belongs to one of the train/valid/test sets.

For all datasets, we encode node text (title and abstract/description) into 1024-dimensional embeddings using BGE-large-en-v1.5~\cite{bge_embedding} as node features.

\partitle{Dynamic Setting}
In the streaming evaluation, we consider the edge insertion case. The graph starts with only train-train edges. Remaining edges are inserted one by one in random order. Every insertion has an independent probability $p$ to trigger retraining (default $p$ = $0.5\%$).

\partitle{Privacy and Model Configuration}
Unless otherwise stated, we use $\delta=10^{-5}$ for all experiments.
We set $\epsilon_{\mathrm{PMA}}=6.0$ with $K=2$ hops, calibrating the
noise scale $\sigma$ to this budget with the R\'enyi-DP accountant of
GAP~\cite{sajadmanesh2023gap}.
The GNN module uses a 3-layer MLP encoder (hidden dim = 256, encoding dim = 128, dropout = 0.5) pre-trained for 100 epochs, followed by a 2-layer MLP classification head trained for 200 epochs. The LLM module uses Llama-3-8B-Instruct~\cite{llama3modelcard} for logit scoring.

\partitle{Baselines}
For node classification, we compare PriDyG against the following methods:
\textbf{ProGAP}~\cite{sajadmanesh2024progap}: the state-of-the-art static
edge-level DP GNN and a stronger static backbone than the
DP GNN pipeline inside PriDyG; we adopt it so that our comparisons
are against the stronger baseline. In the dynamic setting, it is
retrained from scratch at each retrain, with $\epsilon$ accumulating
under sequential composition (trivial retrain) or with geometrically
decaying budgets.
\textbf{LLM-only}: The ablation baseline uses LLM logit scoring using only node text, without any graph-based component.
\textbf{PriDyG w/o LLM}: The ablation baseline uses Graph score without any LLM-based component.
\textbf{Other non-private baselines}: results on the same datasets as reported in~\citep{sun2025graphicl}, marked with `*'.

For link prediction, we compare these baseline methods:
\textbf{Cosine}: GNN-based node embeddings with cosine similarity scoring.
\textbf{LLM-only}: LLM logit scoring using only node text, without any graph-based component.
\textbf{DPLP~\cite{ran2024differentially}}: A private-training link prediction method that relies on query-specific subgraph extraction, which does not protect the queried graph structure at inference time.

\subsection{Node Classification Results}

\noindent\textbf{Static Evaluation.}
Table~\ref{tab:static_nc} reports node classification accuracy in the static setting, where the full graph is available, and the model is trained once. PriDyG achieves accuracy comparable to the ProGAP baseline and other non-DP baselines. 

\begin{table}[t]
\caption{Static node classification accuracy.}
\label{tab:static_nc}
\centering
\scriptsize
\begin{tabular}{lcccccc}
\toprule
\textbf{Method} & \textbf{Cora} & \textbf{PubMed} & \textbf{Arxiv} & \textbf{Products} & $\epsilon$ & $\delta$ \\
\midrule

LLM-only        & 0.613 & 0.887 & 0.520 & 0.631 & 0 & 0\\
GCN* & 0.751 & 0.864 & 0.726 & 0.756 & N/A & N/A\\
LLaGA-HO*~\cite{chen2024llaga} & 0.886 & 0.888 & 0.740 & 0.736 & N/A & N/A\\
GraphGPT*~\cite{tang2024graphgpt} & N/A & 0.847 & 0.622 & N/A & N/A & N/A\\
GraphICL-70B*\cite{sun2025graphicl} & 0.836 & 0.931 & 0.737 & 0.815 & N/A & N/A \\
ProGAP  & 0.725 & 0.914 & 0.729 & 0.859 & 6.0 & $10^{-5}$ \\
PriDyG w/o LLM & 0.722 & 0.914 & 0.726 & 0.850 & 6.0 & $10^{-5}$ \\
PriDyG    & 0.742 & 0.915 & 0.734 & 0.872 & 6.0 & $10^{-5}$ \\
\bottomrule
\end{tabular}
\end{table}

\noindent\textbf{Dynamic Evaluation: Edge Insertion.}
\begin{table}[t]
\caption{Node classification accuracy under edge insertion.}
\label{tab:edge_insertion_nc}
\centering
\small
\resizebox{1\columnwidth}{!}{%
\begin{tabular}{lccccccccc}
\toprule
& \multicolumn{2}{c}{\textbf{Cora}} & \multicolumn{2}{c}{\textbf{PubMed}} & \multicolumn{2}{c}{\textbf{Arxiv}} & \multicolumn{2}{c}{\textbf{Products}} \\
\cmidrule(lr){2-3} \cmidrule(lr){4-5} \cmidrule(lr){6-7} \cmidrule(lr){8-9}
\textbf{Method} & Acc & $\epsilon$ & Acc & $\epsilon$ & Acc & $\epsilon$ & Acc & $\epsilon$ \\
\midrule
ProGAP + Without retrain & 0.709 & 6 & 0.908 & 6 & 0.726 & 6 & 0.860 & 6  \\

ProGAP + Trivial retrain             & 0.722 & 168  & 0.917 & 924   & 0.759 & 12,372  & 0.875 & 18,192 \\

ProGAP + Geometric Decay & 0.703 & 6 & 0.878 & 6 & 0.708 & 6 & 0.856 & 6\\

LLM-only        & 0.621 & 0    & 0.900 & 0     & 0.543 & 0      & 0.614 & 0 \\

PriDyG w/o LLM & 0.719 & 6 & 0.914 & 6 & 0.740 & 6 &0.873 & 6 \\

PriDyG   & \textbf{0.738} & 6  & \textbf{0.918} & 6   & \textbf{0.760} & 6    & \textbf{0.881} & 6 \\
\bottomrule
\end{tabular}
}
\end{table}
Table~\ref{tab:edge_insertion_nc} reports results under the edge
insertion protocol: edges arrive sequentially, and at each insertion
we query their endpoints that belong to the test set. 
PriDyG achieves the highest accuracy among all constant-budget methods, and matches naive retraining, which consumes up to three orders of magnitude more privacy budget. Most of this gain comes from the GNN branch (PriDyG w/o LLM), yet the LLM branch---despite its lower standalone accuracy---still contributes a consistent improvement, since the two branches are accurate on complementary label sets.
Among the constant-budget
alternatives, geometric decay performs worst---falling below even
ProGAP without retraining, i.e., a model frozen at $\mathcal{G}^{(0)}$
that ignores all subsequent edges. This ordering is intuitive:
geometric decay pays for freshness with exponentially growing noise, which ultimately destroys the graph signal and negates the benefit of incorporating the new edges. In contrast, the frozen model at least retains a clean snapshot.

Accuracies here exceed those in Table~\ref{tab:static_nc} because the query sets differ: the streaming protocol queries only endpoints of inserted edges, which skew
toward high-degree nodes, whereas the static setting also includes
isolated and low-degree test nodes. This selection effect inflates
all methods uniformly, so comparisons should be read within each
table; Table~\ref{tab:full-retrain-nc}, which evaluates the full test set
after the stream ends, is the like-for-like counterpart to
Table~\ref{tab:static_nc}.

\begin{table}[ht]
    \centering
    \caption{Node classification accuracy after edge insertion.}
    \scriptsize
    \begin{tabular}{lccccc}
    \toprule
         \textbf{Method} &  \textbf{Cora} & \textbf{PubMed} & \textbf{Arxiv} & \textbf{Product}\\
         \midrule
        PriDyG w/o LLM & 0.703 & 0.896& 0.705 & 0.842\\
        PriDyG & 0.731 & 0.915 & 0.721 & 0.851 \\
    \bottomrule
    \end{tabular}
    \label{tab:full-retrain-nc}
\end{table}

\noindent\textbf{Long-term effectiveness.} Table~\ref{tab:full-retrain-nc} reports the accuracy of PriDyG on the full test set after all the edge insertions. 
PriDyG loses only $1\%$ on average relative to static training on the same graph, against around $2\%$ for its GNN branch alone, indicating that the LLM-GNN collaboration remains effective under long-term updates. 

\begin{table}[ht]
    \centering
    \caption{Effect of privacy budgets $\epsilon$ on node classification.}
    \label{tab:different-epsilon-nc}
    \scriptsize
    \begin{tabular}{lcccccc}
        \toprule
        \multirow{2}{*}{Dataset}
        & \multicolumn{3}{c}{ProGAP}
        & \multicolumn{3}{c}{PriDyG} \\
        \cmidrule(lr){2-4}
        \cmidrule(lr){5-7}
        & $\epsilon=1$ & $\epsilon=6$ & $\epsilon=10$
        & $\epsilon=1$ & $\epsilon=6$ & $\epsilon=10$ \\
        \midrule
        Cora
        & 0.712 & 0.725 & 0.723
        & 0.731 & 0.742 & 0.739 \\
        PubMed
        & 0.904 & 0.914 & 0.915
        & 0.907 & 0.915 &0.927 \\
        Arxiv
        & 0.711 & 0.729 & 0.727
        & 0.730 &  0.734 & 0.745 \\
        Products
        & 0.847 & 0.859 & 0.859
        & 0.853 & 0.872 & 0.874 \\
        \bottomrule
    \end{tabular}
\end{table}

\noindent\textbf{Effect of the privacy budget.} 
Table \ref{tab:different-epsilon-nc} shows the accuracy of PriDyG under different $\epsilon$ for static node classification. The accuracy increases noticeably from $\epsilon=1$ to $\epsilon=6$, while the increase from $\epsilon=6$ to $\epsilon=10$ is small, suggesting that the high $ \epsilon$ has reached the saturation point. PriDyG consistently outperforms or matches ProGAP across privacy budgets.

\subsection{Link Prediction Results}

\noindent\textbf{Static Evaluation. }
Table~\ref{tab:static_lp} reports link prediction performance in the static setting. We use AUC as the evaluation metric, i.e., the model's ability to separate true edges from randomly sampled negative (non-existent) edges by predicted scores. 
PriDyG achieves the highest AUC on all four datasets, outperforming both the graph-only cosine baseline and LLM-only scoring. 
The LLM supplies semantic similarity that compensates for DP noise in the embeddings at zero additional privacy cost. DPLP does not complete training on Arxiv or Products within 72 hours, reported as `N/A'.

\begin{table}[t]
\caption{Static link prediction AUC ($\epsilon = 6.0$, $\delta = 10^{-5}$).}
\label{tab:static_lp}
\centering
\scriptsize
\begin{tabular}{lcccc}
\toprule
\textbf{Method} & \textbf{Cora} & \textbf{PubMed} & \textbf{Arxiv} & \textbf{Products} \\
\midrule
Cosine      & 0.888 & 0.853 & 0.961 & 0.976 \\
DPLP  & 0.910 & 0.949 & N/A & N/A \\
LLM-only  & 0.943 & 0.965 & 0.983 & 0.975 \\
PriDyG       & \textbf{0.947} & \textbf{0.977} & \textbf{0.989} & \textbf{0.986} \\
\bottomrule
\end{tabular}
\end{table}

\noindent\textbf{Dynamic Evaluation: Edge Insertion.}
Table~\ref{tab:edge_insertion_lp} compares link prediction AUC under edge insertion.  PriDyG achieves the best AUC with the smallest budget cost across all four datasets.

\begin{table}[t]
\caption{Link prediction AUC under edge insertion.}
\label{tab:edge_insertion_lp}
\centering
\small
\resizebox{1\columnwidth}{!}{%
\begin{tabular}{lccccccccc}
\toprule
& \multicolumn{2}{c}{\textbf{Cora}} & \multicolumn{2}{c}{\textbf{PubMed}} & \multicolumn{2}{c}{\textbf{Arxiv}} & \multicolumn{2}{c}{\textbf{Products}} \\
\cmidrule(lr){2-3} \cmidrule(lr){4-5} \cmidrule(lr){6-7} \cmidrule(lr){8-9}
\textbf{Method} & AUC & $\epsilon$ & AUC & $\epsilon$ & AUC & $\epsilon$ & AUC & $\epsilon$ \\
\midrule
ProGAP + Without retrain & 0.840 & 6 & 0.783 & 6 & 0.910 & 6 & 0.921 & 6 \\

ProGAP + Trivial retrain  & 0.852 & 150 & 0.787 & 732 & 0.923 & 9,828 & 0.939 & 14,286 \\
ProGAP + Geometrically Decaying & 0.824 & 6& 0.774 & 6 & 0.893 & 6 & 0.902 & 6\\
LLM-only & 0.916  & 0    & 0.961 & 0     & 0.975  & 0  & 0.972  & 0 \\
PriDyG w/o LLM & 0.852 & 6 & 0.786 & 6 & 0.924 & 6 & 0.936 & 6\\
PriDyG  & \textbf{0.929} & 6 & \textbf{0.967} & 6 & \textbf{0.983} & 6 & \textbf{0.978} & 6\\
\bottomrule
\end{tabular}
}
\end{table}

\noindent\textbf{Long-term effectiveness.} Table \ref{tab:full-retrain-lp} reports the accuracy of PriDyG on the test set after all the edge insertions. As in the node classification task, compared with fully training, PriDyG's accuracy drops by only roughly $1\%$, while the GNN part drops by at most more than $5\%$, demonstrating PriDyG's effectiveness under long-term updates. 

\begin{table}[ht]
    \centering
    \caption{Link prediction AUC after edge insertion.}
    \scriptsize
    \begin{tabular}{lccccc}
    \toprule
         \textbf{Method} &  \textbf{Cora} & \textbf{PubMed} & \textbf{Arxiv} & \textbf{Product}\\
         \midrule
          PriDyG w/o LLM & 0.855 & 0.793& 0.927 & 0.934\\
        PriDyG & 0.938 & 0.964 & 0.977 & 0.980 \\
    \bottomrule
    \end{tabular}
    \label{tab:full-retrain-lp}
\end{table}

\begin{table}[ht]
    \centering
    \caption{Effect of the privacy budgets $\epsilon$ on link prediction.}
    \label{tab:different-epsilon-lp}
    \scriptsize
    \begin{tabular}{lcccccc}
        \toprule
        \multirow{2}{*}{Dataset}
        & \multicolumn{3}{c}{Cosine}
        & \multicolumn{3}{c}{PriDyG} \\
        \cmidrule(lr){2-4}
        \cmidrule(lr){5-7}
        & $\epsilon=1$ & $\epsilon=6$ & $\epsilon=10$
        & $\epsilon=1$ & $\epsilon=6$ & $\epsilon=10$ \\
        \midrule
        Cora
        & 0.860 & 0.888 & 0.897
        & 0.935 & 0.947 & 0.948 \\
        PubMed
        & 0.808 & 0.853 & 0.855
        & 0.954 & 0.977 & 0.969 \\
        Arxiv
        & 0.954 & 0.961 & 0.959
        & 0.988 & 0.989 & 0.989 \\
        Products
        & 0.967 & 0.976 & 0.977
        & 0.983 & 0.986 & 0.986 \\
        \bottomrule
    \end{tabular}
\end{table}

\noindent\textbf{Effect of the privacy budget.} In Table \ref{tab:different-epsilon-lp}, we show the accuracy of PriDyG under different $\epsilon$ for static link prediction. The conclusion is similar to that of Table \ref{tab:different-epsilon-nc}: the AUC increases noticeably from $\epsilon=1$ to $\epsilon = 6$. The AUC difference between $\epsilon =6$ and $\epsilon=10$ is quite small.  PriDyG also consistently outperforms ProGAP across privacy budgets, with a significant margin in most cases.

\section{Related Works}

\partitle{Differentially Private Graph Learning}
Differential privacy~\cite{dwork2014algorithmic} is well established for deep learning~\cite{abadi2016deep} and increasingly applied to graph-structured data~\cite{mueller2022sok}. Under edge-level DP, existing works perturb node features before aggregation~\cite{sajadmanesh2021locally}, decouple graph convolutions~\cite{zhang2024dpar,chien2023differentially}, or protect edges specifically for link prediction~\cite{ran2024differentially,qi2024linkguard}. Most relevant to us, GAP~\cite{sajadmanesh2023gap} adds calibrated Gaussian noise to multi-hop aggregation, and ProGAP~\cite{sajadmanesh2024progap} extends it with progressive training. All of them, however, assume static graphs and a single round of training. When the graph evolves over time, and the model must be periodically retrained, sequential composition causes the privacy budget to grow linearly with the number of retrains. We close this gap with incremental PMA, which exploits parallel composition over disjoint edge batches to achieve constant total privacy cost.

\partitle{Privacy-Preserving Dynamic Graph Analysis}
Prior work on dynamic graphs under edge-level DP targets traditional graph queries such as graph density~\cite{paul2023edge,raskhodnikova2025fully}, or graph \textit{publishing}---e.g., DyGAN-EDP~\cite{ma2026novel} models a dynamic graph as discrete snapshots and synthesizes privacy-preserving future ones with a GAN. To the best of our knowledge, PriDyG is the first framework for edge-level DP inference over continually updated graphs at a constant privacy budget.

\partitle{LLMs for Graph Learning}
LLMs have shown strong potential for graph tasks by leveraging node-level textual attributes~\cite{chen2024exploring,wang2024llms}, through graph-to-token mappings~\cite{chen2024llaga}, graph instruction tuning~\cite{tang2024graphgpt}, or structured prompting for in-context learning~\cite{sun2025graphicl}. None of them consider privacy constraints. PriDyG instead exploits the fact that LLM inference relies solely on node text, which edge-level DP does not protect: the LLM branch incurs zero privacy cost, making it an ideal complement to the DP-constrained GNN.

\section{Conclusion}

We formulated the Edge-level DP Dynamic Graph Inference (EDG) problem and proposed PriDyG, which combines GNN-based structural learning with LLM-based semantic reasoning under edge-level DP. 
PriDyG introduces two contributions: (1) incremental PMA, which exploits the linearity of graph aggregation and disjoint edge batching to keep the total privacy budget independent of the number of model updates, and (2) privacy-free LLM augmentation, which uses node text at zero additional privacy cost to complement the DP-constrained GNN through confidence-gated fusion. 
Experiments on four benchmarks show that PriDyG attains utility with non-private baselines on both node classification and link prediction, while reducing cumulative privacy cost by orders of magnitude over naive retraining. Future work includes extending PriDyG to node-level privacy and to graph streams with edge deletions.

\bibliographystyle{ACM-Reference-Format}
\bibliography{reference} 

\appendix
\section{Proofs}
\label{appendix:proofs}

\begin{proof}[\textbf{Proof for Lemma \ref{lem:rr}}]
Neighboring graphs differ in one undirected edge
$e^\star=\{u^\star,v^\star\}$. Each pair's bit is perturbed by
Eq.~\eqref{eq:rr} at most once, and the likelihood ratio of the
two-point RR channel is bounded by $e^{\varepsilon_{\mathrm{RR}}}$,
giving pure $\varepsilon_{\mathrm{RR}}$-DP for the calibration
release. For the joint statement: since $h$ is graph-independent,
$e^\star$ lies in exactly one of $\mathcal{P}_{\mathrm{train}}$,
$\mathcal{P}_{\mathrm{fus}}$. If
$e^\star\in\mathcal{P}_{\mathrm{train}}$, the calibration bits are
identically distributed under both graphs and the PMA transcript is
$(\varepsilon_{\mathrm{PMA}},\delta)$-DP by Theorem~\ref{thm:constant-budget}; if
$e^\star\in\mathcal{P}_{\mathrm{fus}}$, the PMA transcript is
identically distributed (the pair is excluded from all adjacency
matrices) and the calibration release is
$(\varepsilon_{\mathrm{RR}},0)$-DP. Parallel composition over the
disjoint partition yields the maximum of the two guarantees.
\end{proof}

\begin{proof}[\textbf{Proof for Theorem \ref{thm:constant-budget}}]
Since each edge enters the buffer exactly once and is drained exactly
once, the batches $\{B_m\}_{m=0}^{M}$ are disjoint by construction.
Consider two neighboring dynamic graphs $D$ and $D'$ that differ in one
undirected edge $e^\star=\{u^\star,v^\star\}$, and let $m^\star$ be the
unique retrain with $e^\star\in B_{m^\star}$. We decompose the
transcript into three segments and compare their conditional
distributions under $D$ and $D'$.
 
\emph{Before $m^\star$.} For $m<m^\star$, the computation accesses only
batches disjoint from $\{e^\star\}$, so $(Y_0,\dots,Y_{m^\star-1})$ has
identical distribution under $D$ and $D'$.
 
\emph{At $m^\star$.} All features consumed at retrain $m^\star$---the
cached aggregations carried over from previous retrains, and the
lower-hop features $\bar{\mathbf{Z}}^{(k-1)}_{\mathrm{new}}$ finalized
earlier within the same retrain---are deterministic functions of the
transcript released so far. Conditional on that transcript, the only
fresh access to $B_{m^\star}$ at hop $k$ is the multiplication by
$\Delta\mathbf{A}^{(m^\star)}$ in Eq.~\eqref{eq:delta-agg}, whose
$\ell_2$-sensitivity to one undirected edge is $\sqrt{2}$: removing
$\{u^\star,v^\star\}$ changes the rows of $u^\star$ and $v^\star$, each
by at most one unit-norm vector. Hence $Y_{m^\star}$, conditional on
the prefix, is a sequence of $K$ adaptively composed Gaussian
mechanisms with sensitivity $\sqrt{2}$ and scale $\sigma$---exactly the
mechanism structure of one static $K$-hop PMA invocation with respect
to the edge $e^\star$.
 
\emph{After $m^\star$.} For $m>m^\star$, the raw edge $e^\star$ is
never re-accessed: later retrains touch only batches disjoint from
$\{e^\star\}$ and otherwise depend on $e^\star$ solely through the
already-released transcript. These segments are therefore (randomized)
post-processing with respect to $e^\star$ and incur no additional
privacy loss.
 
Consequently, for any pair of neighboring graphs, the privacy loss of
the full transcript is exactly that of a single static $K$-hop PMA
invocation with noise scale $\sigma$. Any
$(\epsilon_{\mathrm{PMA}},\delta)$-DP guarantee certified for static
PMA therefore transfers verbatim to $Y_{0:M}$, for every $M$ and every
retraining schedule.
\end{proof}

\begin{proof}[\textbf{Proof of Theorem~\ref{thm:captured-paths}}]
Let $\mathbf{I}^{(k)}_m$ denote the noiseless hop-$k$ cache after
retrain $m$. By Steps~(2)--(3) of the algorithm with the noise
omitted, the caches obey
\begin{equation}
\mathbf{I}^{(k)}_m
=\mathbf{I}^{(k)}_{m-1}+\mathbf{A}_m\,\mathbf{I}^{(k-1)}_m,
\qquad
\mathbf{I}^{(0)}_m=\mathbf{Z}\ \ \text{for all }m,
\label{eq:noiseless-recursion}
\end{equation}
since $\bar{\mathbf{Z}}^{(0)}$ depends only on node text, and hop $k$
multiplies the delta adjacency by the \emph{finalized} hop-$(k{-}1)$
features of the current retrain.
 
We prove \eqref{eq:captured-paths} by induction on $(m,k)$ in
lexicographic order. For $m=0$, the initial retrain runs full PMA on
$B_0$, so $\mathbf{I}^{(k)}_0=\mathbf{A}_0^{\,k}\mathbf{Z}$, the
unique (trivially non-decreasing) length-$k$ word over $\{0\}$. For
$k=0$, both sides equal $\mathbf{Z}$ (the empty path). For the
inductive step, assume the claim for $(m-1,k)$ and $(m,k-1)$. In
\eqref{eq:noiseless-recursion}, the first term is the sum over all
non-decreasing paths with indices in $\{0,\dots,m-1\}$, i.e., those
with $i_k\le m-1$; the second term equals
\[
\mathbf{A}_m\!\!\sum_{0\le i_1\le\cdots\le i_{k-1}\le m}\!\!
\mathbf{A}_{i_{k-1}}\cdots\mathbf{A}_{i_1}\mathbf{Z}
\;=\!\!\sum_{\substack{0\le i_1\le\cdots\le i_k\le m\\ i_k=m}}\!\!
\mathbf{A}_{i_k}\cdots\mathbf{A}_{i_1}\mathbf{Z},
\]
exactly the non-decreasing paths with $i_k=m$ (the constraint
$i_{k-1}\le m$ is automatic). The two index sets are disjoint and
their union is the set of all non-decreasing paths over
$\{0,\dots,m\}$, proving \eqref{eq:captured-paths}. Expanding
$\mathbf{F}^{(k)}=\bigl(\sum_{i\le m}\mathbf{A}_i\bigr)^{k}\mathbf{Z}$
yields all length-$k$ paths.
 
For a single retrain
($\mathbf{A}_{\mathrm{new}}=\mathbf{A}_{\mathrm{old}}+\Delta\mathbf{A}$,
i.e., $m=1$ with $\mathbf{A}_0=\mathbf{A}_{\mathrm{old}}$ and
$\mathbf{A}_1=\Delta\mathbf{A}$), \eqref{eq:captured-paths}
specializes to
\begin{equation}
\mathbf{I}^{(k)}
=\sum_{i=0}^{k}\Delta\mathbf{A}^{\,i}\,
\mathbf{A}_{\mathrm{old}}^{\,k-i}\,\mathbf{Z},
\label{eq:inc-unrolled}
\end{equation}
the paths in which every old-edge step precedes every new-edge step.
Writing $\mathbf{C}^{(k)}=\mathbf{A}_{\mathrm{old}}^{\,k}\mathbf{Z}$
for the cached aggregation, the error
$\mathbf{D}^{(k)}=\mathbf{F}^{(k)}-\mathbf{I}^{(k)}$ collects all
remaining paths---those containing at least one old-edge step
\emph{after} a new-edge step---and obeys the recursion
\begin{equation}
\mathbf{D}^{(k)}
=\mathbf{A}_{\mathrm{old}}\bigl(\mathbf{F}^{(k-1)}-\mathbf{C}^{(k-1)}\bigr)
+\Delta\mathbf{A}\,\mathbf{D}^{(k-1)},
\qquad \mathbf{D}^{(1)}=\mathbf{0}.
\label{eq:err-recursion}
\end{equation}
The first term of \eqref{eq:err-recursion} exposes the source of the
error: it is $\mathbf{A}_{\mathrm{old}}$ propagating the staleness of
the cache---new edges always aggregate up-to-date features, but old
edges, frozen in the cache, never re-aggregate features that were
refreshed by later insertions.
\end{proof}
 
\begin{proof}[\textbf{Proof of Corollary~\ref{cor:hop1-exact}}]
For $k=1$, every length-$1$ word is trivially non-decreasing, so by
Theorem~\ref{thm:captured-paths},
$\mathbf{I}^{(1)}=\sum_{b\le m}\mathbf{A}_b\,\mathbf{Z}
=\mathbf{A}_{\mathrm{new}}\mathbf{Z}=\mathbf{F}^{(1)}$.
\end{proof}
 
\begin{proof}[\textbf{Proof of Corollary~\ref{cor:hop2-capture}}]
A wedge with edges $e_1,e_2$ carries two transfers, corresponding to
the paths $\bigl(e_1,e_2\bigr)$ and $\bigl(e_2,e_1\bigr)$.
If the two edges were inserted in the same batch, or both lie in $\mathcal{E}^{(0)}$---then both
paths are non-decreasing, and both transfers are captured by
Theorem~\ref{thm:captured-paths}. If the two edges were inserted in different batches, then one edge
is necessarily older than the other and exactly one path (older
edge first) is non-decreasing: $50\%$. Hence the captured fraction is
$(1-\rho)\cdot 1+\rho\cdot\tfrac12=1-\tfrac12\rho\ge\tfrac12$. For a
single retrain, this recovers
$\mathbf{D}^{(2)}=\mathbf{A}_{\mathrm{old}}\,\Delta\mathbf{A}\,\mathbf{Z}$
from \eqref{eq:err-recursion}: the traversal crossing the old edge
first is captured
($\Delta\mathbf{A}\,\mathbf{A}_{\mathrm{old}}\mathbf{Z}$ appears in
\eqref{eq:inc-unrolled}), while the traversal crossing the new edge
first is lost.
\end{proof}
 
\begin{proof}[\textbf{Proof of Proposition~\ref{prop:noise-accum}}]
The initial retrain injects one draw of
$\mathcal{N}(\mathbf{0},\sigma^2\mathbf{I})$ per hop (Eq.~(4)), and
each of the $M$ subsequent retrains injects one fresh draw per hop
via Eq.~(8), accumulated additively into the cache by Eq.~(9). The
$M{+}1$ draws are mutually independent, so their variances sum to
$(M{+}1)\sigma^2$. By Theorem~\ref{thm:constant-budget}, parallel
composition over disjoint batches lets every incremental retrain use
the full per-invocation noise scale $\sigma$; a one-shot full PMA on
the final graph at the same
$(\epsilon_{\mathrm{PMA}},\delta)$ calibrates to the \emph{same}
$\sigma$ but perturbs each hop with a single draw, giving the
$\sqrt{M{+}1}$ inflation in noise magnitude at identical privacy
cost. At hop $1$, the noiseless aggregations coincide by
Corollary~\ref{cor:hop1-exact}, so the deviation from the one-shot
release is exactly this injected noise. For the geometrically
decaying baseline, the Gaussian noise scale required for budget
$\epsilon_m=(1-q)q^m\epsilon_{\mathrm{tot}}$ at fixed $\delta$
satisfies $\sigma_m\propto 1/\epsilon_m
=q^{-m}/\bigl((1-q)\epsilon_{\mathrm{tot}}\bigr)
=\Theta\!\bigl(q^{-m}\bigr)$; and because that baseline re-runs full
PMA, the entire representation released at retrain $m$ is perturbed
at scale $\sigma_m$, so utility collapses exponentially fast in the
number of retrains.
\end{proof}

\section{LLM Prompt Templates}
\label{appendix:prompts}

This appendix documents all LLM prompt templates used in our experiments. We use zero-shot prompting with logit-based scoring: the LLM computes log-probabilities over candidate answer tokens (e.g., category names or ``Yes''/``No''), and the candidate with the highest probability is selected. All prompts follow a consistent structure tailored to each task and dataset.

\subsection{Node Classification Prompts}

\begin{tcolorbox}[colback=gray!5, colframe=gray!50, title=Cora Node Classification Prompt]
\small
\begin{verbatim}
You are a scientific paper classification assistant.
Choose the single best paper category for this paper.
Paper title:
{title}

Paper abstract:
{abstract}

Choose the single best paper category for this paper
from the candidate list.
Candidates:
 Case Based: Papers on case-based reasoning,
   instance-based learning, and analogy-driven
   problem solving.
 Genetic Algorithms: Papers on evolutionary
   computation, genetic programming, and optimization
   via natural selection.
 Neural Networks: Papers on artificial neural
   networks, deep learning, backpropagation, and
   connectionist models.
 Probabilistic Methods: Papers on Bayesian methods,
   graphical models, expectation-maximization, and
   statistical inference.
 Reinforcement Learning: Papers on reward-based
   learning, policy optimization, Markov decision
   processes, and Q-learning.
 Rule Learning: Papers on inductive logic programming,
   decision rules, association rules, and symbolic
   rule extraction.
 Theory: Papers on computational learning theory,
   PAC learning, VC dimension, and formal analysis
   of algorithms.

Output ONLY the category name from the candidate list.
Do not output anything else.
Answer:
\end{verbatim}
\end{tcolorbox}

\begin{tcolorbox}[colback=gray!5, colframe=gray!50, title=ogbn-arxiv Node Classification Prompt]
\small
\begin{verbatim}
Classify this computer science paper into exactly one
arXiv category.

Title: {title}
Abstract: {abstract}

Categories:
 0: cs.AI (Artificial Intelligence)
 1: cs.AR (Hardware Architecture)
 2: cs.CC (Computational Complexity)
 ...
 39: cs.SY (Systems and Control)

Output ONLY the integer label id from the candidate
list. Do not output anything else.
Answer:
\end{verbatim}
\end{tcolorbox}

The full list of 40 arXiv CS subcategories used as candidates is shown in Table~\ref{tab:arxiv-categories}.

\begin{table}[h]
\centering
\caption{arXiv CS subcategories used as candidate labels.}
\label{tab:arxiv-categories}
\small
\resizebox{\columnwidth}{!}{%
\begin{tabular}{ll|ll}
\toprule
\textbf{Code} & \textbf{Description} & \textbf{Code} & \textbf{Description} \\
\midrule
cs.AI & Artificial Intelligence & cs.IT & Information Theory \\
cs.AR & Hardware Architecture & cs.LG & Machine Learning \\
cs.CC & Computational Complexity & cs.LO & Logic in Computer Science \\
cs.CE & Comp.\ Eng., Finance, \& Sci. & cs.MA & Multiagent Systems \\
cs.CG & Computational Geometry & cs.MM & Multimedia \\
cs.CL & Computation and Language & cs.MS & Mathematical Software \\
cs.CR & Cryptography and Security & cs.NA & Numerical Analysis \\
cs.CV & Computer Vision \& Pattern Rec. & cs.NE & Neural \& Evolutionary Comp. \\
cs.CY & Computers and Society & cs.NI & Networking \& Internet Arch. \\
cs.DB & Databases & cs.OH & Other Computer Science \\
cs.DC & Distributed \& Parallel Comp. & cs.OS & Operating Systems \\
cs.DL & Digital Libraries & cs.PF & Performance \\
cs.DM & Discrete Mathematics & cs.PL & Programming Languages \\
cs.DS & Data Structures \& Algorithms & cs.RO & Robotics \\
cs.ET & Emerging Technologies & cs.SC & Symbolic Computation \\
cs.FL & Formal Lang.\ \& Automata Th. & cs.SD & Sound \\
cs.GL & General Literature & cs.SE & Software Engineering \\
cs.GR & Graphics & cs.SI & Social \& Info.\ Networks \\
cs.GT & CS and Game Theory & cs.SY & Systems and Control \\
cs.HC & Human-Computer Interaction & & \\
cs.IR & Information Retrieval & & \\
\bottomrule
\end{tabular}
}
\end{table}

\begin{tcolorbox}[colback=gray!5, colframe=gray!50, title=PubMed Node Classification Prompt]
\small
\begin{verbatim}
You are a biomedical paper classification assistant.
Choose the single best diabetes category for this paper.
Paper title:
{title}

Paper abstract:
{abstract}

Choose the single best diabetes category for this paper
from the candidate list.
Candidates:
 Diabetes Mellitus, Experimental
 Diabetes Mellitus, Type 1
 Diabetes Mellitus, Type 2

Output ONLY the category name from the candidate list.
Do not output anything else.
Answer:
\end{verbatim}
\end{tcolorbox}

\begin{tcolorbox}[colback=gray!5, colframe=gray!50, title=ogbn-products Node Classification Prompt]
\small
\begin{verbatim}
You are a product classification assistant.
Choose the single best top-level product category for
this Amazon product from the candidate list.
Product title:
{title}
Product description:
{content}
Choose the single best top-level product category for
the given product.
Candidates:
0: Home & Kitchen
1: Health & Personal Care
2: Beauty
...

Output ONLY the integer label id from the candidate
list. Do not output anything else.
Answer:
\end{verbatim}
\end{tcolorbox}

\subsection{Link Prediction Prompts}

For link prediction, we use a unified prompt template parameterized by dataset-specific context. The LLM scores the log-probabilities of ``Yes'' and ``No'' tokens, and we use $P(\text{Yes})$ as the link likelihood score.

\begin{tcolorbox}[colback=gray!5, colframe=gray!50, title=Link Prediction Prompt -- Citation Networks]
\small
\begin{verbatim}
You are an expert in analyzing {context}.
Given the descriptions of two papers below, determine
whether they are likely to cite each other (i.e., one
paper references the other).

Paper A:
{description of paper A}

Paper B:
{description of paper B}

Based on the content similarity and topical relevance,
are these two papers likely to cite each other (i.e.,
one paper references the other)?
Output ONLY 'Yes' or 'No'.
Answer:
\end{verbatim}
\end{tcolorbox}

The \texttt{\{context\}} field is set to:
\begin{itemize}
    \item \textbf{Cora}: ``academic papers in computer science (Cora citation network)''
    \item \textbf{ogbn-arxiv}: ``academic papers from arXiv (OGB-arxiv citation network)''
    \item \textbf{PubMed}: ``biomedical research papers from PubMed (citation network)''
\end{itemize}

\begin{tcolorbox}[colback=gray!5, colframe=gray!50, title=Link Prediction Prompt -- Co-purchasing Network]
\small

\begin{verbatim}
You are an expert in analyzing products from Amazon
(co-purchasing network).
Given the descriptions of two products below, determine
whether they are likely to are frequently co-purchased
together.

Product A:
{description of product A}

Product B:
{description of product B}

Based on the content similarity and topical relevance,
are these two products likely to are frequently
co-purchased together?
Output ONLY 'Yes' or 'No'.
Answer:
\end{verbatim}
\end{tcolorbox}

\end{document}